\documentclass[11pt,a4paper]{article}

\usepackage[mathscr]{eucal}
\usepackage[font=footnotesize]{subfig}
\usepackage{multirow}
\usepackage{url}
\usepackage{makecell}

\usepackage{amsmath, amsfonts, amsthm}
\usepackage{amsopn}
\usepackage{hyperref}
\usepackage[pdftex]{graphicx}
\usepackage{booktabs} % To thicken table lines
\usepackage{authblk}
\usepackage[heightrounded]{geometry}
\usepackage{pdfpages}

\theoremstyle{plain}
\newtheorem{prop}{Proposition}
\theoremstyle{remark}
\newtheorem{remark}{Remark}

\newcommand{\R}{\mathbb{R}}
\newcommand{\Meff}{M_{\text{eff}}}

\newcommand{\A}{\mathscr{A}}
\def \PP{\mathcal{P}}

\DeclareMathOperator{\Id}{Id}
\DeclareMathOperator*{\argmin}{argmin}

\title{Prediction of residue-residue contacts from protein families using similarity kernels and least squares regularization}
\author[1]{Massimo Andreatta}
\author[1]{Santiago Laplagne}
\author[2]{Shuai Cheng Li}
\author[1]{Stephen Smale}
\affil[1]{\footnotesize{Department of Mathematics, City University of Hong Kong, Kowloon, Hong Kong}}
\affil[2]{\footnotesize{Department of Computer Science, City University of Hong Kong, Kowloon, Hong Kong}}
\date{March 26, 2014}

\begin{document}

%\author[Andreatta \textit{et~al}]{Massimo Andreatta\,$^{1*}$, Santiago Laplagne\,$^{1}$\footnote{These authors contributed equally to this work.}\ , Shuai Cheng Li\,$^{2}$ and Stephen Smale\,$^{1}$\footnote{To whom correspondence should be addressed.}}
%\address{$^{1}$Department of Mathematics, City University of Hong Kong, Kowloon, Hong Kong\\
%$^{2}$Department of Computer Science, City University of Hong Kong, Kowloon, Hong Kong}
%\footnote{to whom correspondence should be addressed}
%\history{Received on XXXXX; revised on XXXXX; accepted on XXXXX}

%\editor{Associate Editor: XXXXXXX}

\maketitle

\begin{abstract}

One of the most challenging and long-standing problems in computational biology is the prediction of three-dimensional protein structure from amino acid sequence. 
A promising approach to infer spatial proximity between residues is the study of evolutionary covariance from multiple sequence alignments, 
especially in light of recent algorithmic improvements and the fast growing size of sequence databases. 
In this paper, we present a simple, fast and accurate algorithm for the prediction of residue-residue contacts based on regularized least squares.
The basic assumption is that spatially proximal residues in a protein coevolve to maintain the physicochemical complementarity of the amino acids involved in the contact.
Our regularized inversion of the sample covariance matrix allows the computation of partial correlations between pairs of residues, thereby removing the effect of spurious transitive correlations.
The method also accounts for low number of observations by means of a regularization parameter that depends on the effective number of sequences in the alignment. 
When tested on a set of protein families from Pfam, we found the RLS algorithm to have performance comparable to state-of-the-art methods for contact prediction, while at the same time being faster and conceptually simpler.

The source code and data sets are available at \url{http://cms.dm.uba.ar/Members/slaplagn/software}

\end{abstract}

\section{Introduction}

\label{section:introduction}

A major problem in computational biology is the prediction of the 3D structure of a protein from its amino acid sequence. 
Anfinsen's dogma suggests that, in principle, the amino acid sequence contains enough information to determine the full three-dimensional structure \cite{anfinsen1973principles}. 
However, a few decades on, the mechanisms of protein folding are still not satisfactorily explained \cite{dill2012protein}. 
In particular, the space of possible spatial configurations given a certain amino acid 1D sequence is immense (the ``Levinthal paradox''), yet an unfolded polypeptide chain is driven to its native 3D structure in a finite time, typically milliseconds to seconds, upon shifting to folding conditions \cite{rose2006backbone}.

Such enormous search space poses important challenges to the development of \textit{ab initio} methods for structure prediction. 
Therefore, it is essential to exploit different kinds of information that can help reduce the degrees of freedom in the configurational search space. 
A powerful way of inferring distance constraints is the prediction of residue-residue contacts from multiple sequence alignments (MSA). 
The underlying assumption is that contacting residues coevolve to maintain the physicochemical complementarity of the amino acids involved in the contact. 
That is, if a mutation occurs in one of the contacting residues, the other one is also likely to mutate, lest the fold of the protein may be disrupted. 
Methods based on residue coevolution aim at inferring spatial proximity between residues (contacts) from such signals of correlated mutations (Figure \ref{fig:contacts}).

\begin{figure}[!b]
\centering
  \includegraphics[width=0.45\textwidth]{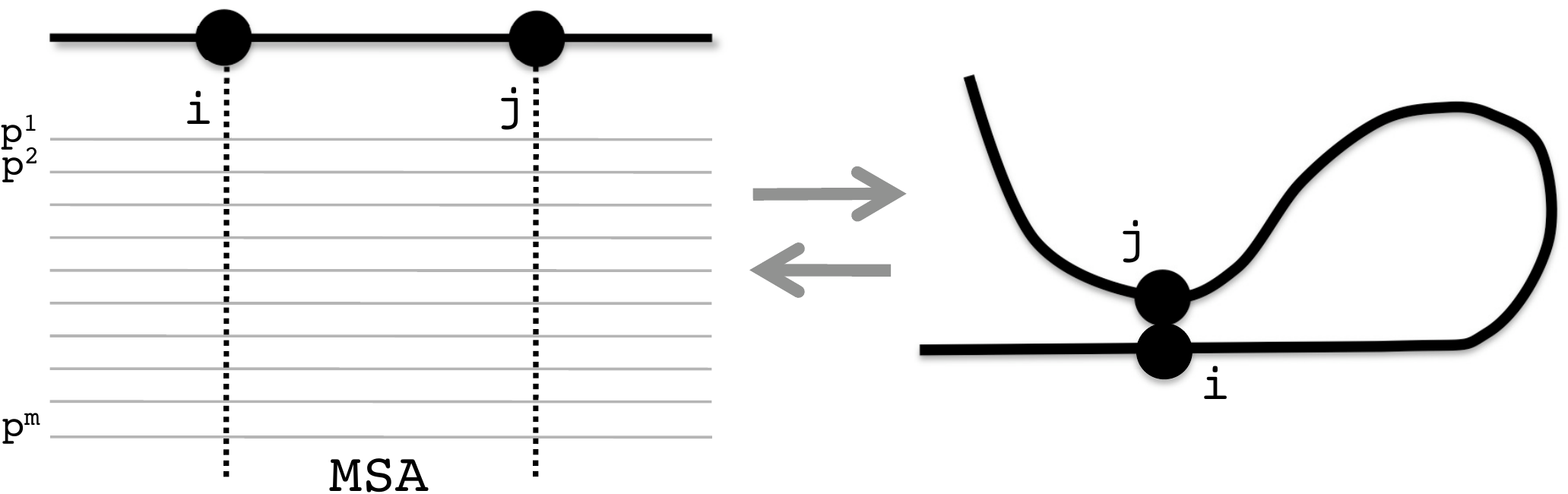}
  \caption{Illustration of a residue-residue contact. The contact imposes a constraint on the evolution of residues $i$ and $j$. Vice versa, coevolution of $i$ and $j$ can be used to infer their physical proximity.}
   \label{fig:contacts}
\end{figure}

Thanks to the recent exponential growth in sequence data collected in databases such as Pfam \cite{PFAM}, algorithms for the prediction of contacting residues from MSA have enjoyed increasing attention. 
Different kinds of approaches have been recently applied for contact prediction, from mutual information (MI) between pairs of positions \cite{buslje2009correction, dunn2008mutual, wang2013predicting}, to Bayesian network models \cite{burger2010disentangling}, direct-coupling analysis \cite{balakrishnan2011learning, morcos2011direct, marks2011protein} and sparse inverse covariance matrix estimation \cite{jones2012psicov}. 
See also \cite{marks2012protein} and \cite{de2013emerging} for recent reviews.
In particular, the more sophisticated and successful methods attempt to disentangle direct and indirect correlations, that is the artifactual correlations emerging from transitive effects of covariance analysis \cite{lapedes1999correlated, weigt2009identification}. 
Morcos et al.\ \cite{morcos2011direct} and Marks et al.\ \cite{marks2011protein} tackle this problem using a maximum-entropy approach, whereas Jones et al.\ \cite{jones2012psicov} estimate partial correlations by inverting the covariance matrix. A very recent pseudo-likelihood method based on 21-state Potts models \cite{ekeberg2013improved} was shown to outperform other approaches for direct-coupling analysis. Kamisetty et al.\ \cite{kamisetty2013assessing} systematically analyzed the conditions under which predicted contacts are likely to be useful for structure prediction, and found several hundred families that meet their criteria.

Here, we propose a new approach for computing direct correlations that employs regularized least squares (RLS) regression to invert a sample covariance matrix $S$. 
We compute the regularized inverse by the formula
\begin{equation}
\label{eq:theta}
 \Theta = (S^2 + \eta \Id)^{-1} S,
\end{equation}
with fixed $\eta > 0$. It proves to be a very simple, direct and fast approach, and requires no assumption on probabilities distributions or sparsity in the correlations.

The RLS algorithm described in this paper was applied to three different sets of protein families, and we compared its performance to state-of-the-art methods for contact prediction. The RLS method achieves precision rates superior to PSICOV \cite{jones2012psicov} and comparable to plmDCA \cite{ekeberg2013improved} but it is considerably faster than either. 

\section{Approach}

\subsection{The covariance matrix} 
\label{section:covariance}

Let $\A$ be the set of $20$ amino acids plus the gap symbol $-$ and $\PP = \{p^m = (p_1^m, \dots, p_L^m)\}_{m=1,\dots,M}$ a given Pfam family of $M$ aligned protein sequences, possibly with gaps, where $L$ denotes the length of the protein domains. On this set of proteins, the covariance between any pair of columns $(i,j)$ for the amino acids pair $(a,b)$ is given by
$$
	S^0_{ij}(a,b) = f_{ij}(a,b) - f_i(a) f_j(b)
$$
where the corrected frequencies are calculated as
\begin{equation}
\label{eq:pseudo}
	f_i (a) = \frac{1}{\lambda + M_{\text{eff}}} \Big( \frac{\lambda}{21} + \sum_p w(p) \delta(a, p_i) \Big)
\end{equation}
\[
	f_{ij} (a,b) = \frac{1}{\lambda + M_{\text{eff}}} \Big( \frac{\lambda}{21^2} + \sum_p w(p) \delta(a, p_i) \delta(b, p_j) \Big)
\]
The delta kernel takes value $\delta(a,b)=1$ if $a = b$ and  $\delta(a,b)=0$ otherwise. $w(p)$ is the weight of protein $p$ and $\Meff = \sum_p w(p)$ (see section \ref{section:measure} for details on sequence weighting).

The parameter $\lambda$ is the so-called pseudocount, a regularization parameter that accounts for non-observed pairs. We note that the same, or similar, constructions for the corrected amino acid frequencies have been proposed previously by other authors \cite{ekeberg2013improved, jones2012psicov, morcos2011direct}.

\subsubsection{Modified covariance matrix}

We set $S^0_{ii}(a,b) = 0$ for $a \neq b$, and call $S$ this new matrix. This modification also appears in the code of PSICOV \cite{jones2012psicov} although it is not stated in their paper. By setting those values to $0$, the resulting matrix contains in general negative eigenvalues (see Figures \eqref{fig:eig_0_00028} and \eqref{fig:eig_2_00028}) and hence is not anymore semi-definite positive, but it is still symmetric. We do not fully understand this step, but it is noteworthy that Equation \ref{eq:theta} still makes sense for any $\eta > 0$.

In general, working with $S$ instead of $S^0$ gives better results in our experiments. See Table S1 for the effect of this step on predictive performance.

\subsection{Regularized inverse -- the key algorithm} As we mentioned in the Introduction, the covariance between our random variables does not distinguish between direct and indirect correlations. 
To overcome this problem, a technique used by statisticians is to compute the so-called partial correlations, which can be obtained from the inverse of the covariance matrix using its associated correlation matrix.

Since the covariance matrix is usually singular or ill conditioned, regularization techniques must be used to compute a regularized inverse $\Theta$.  We achieve this by solving the following optimization problem
\begin{equation}
\label{eq:optimization}
 \Theta = \argmin_{X \in \R^{20L \times 20L}} \| SX - \Id \|_2^2 + \eta \|X\|_2^2,
\end{equation}
where $\| \cdot \|_2$ denotes the Frobenius norm, and $\eta$ is a regularization parameter to be determined. Observe that the first term is minimized by the inverse of $S$ when it exists.

The problem has a unique solution for any $\eta > 0$ as we see in the next proposition.

\begin{prop}
\label{prop:regularization}
For a symmetric matrix $S \in \R^{n \times n}$ and a regularization parameter $\eta > 0$, the optimization problem \eqref{eq:optimization}
has a unique solution, which is also symmetric and given by equation \ref{eq:theta}. When $S$ is semidefinite positive, then the solution also is.
\end{prop}

\begin{proof}
Since the norms involved are coordinate norms, the problem can be decoupled into independent problems for each column of $X$:
\[
\Theta^{(i)} = \argmin_{x \in \R^{n \times 1}} \| S^tx - e^{(i)} \|_2^2 + \eta \|x\|_2^2,
\]
where $\Theta^{(i)}$ is the $i$-th column of $\Theta$ and $e^{(i)}$ is the $i$-th column of the identity matrix.

This is a well studied problem known as regularized least squares (also called Tikhonov regularization or Ridge regression in different areas, see \cite{tikhonov1943stability} and \cite{hoerl1962application}).
The unique solution is $\Theta^{(i)} = (S^tS + \eta \Id)^{-1} S^t e^{(i)}$.

Hence, the solution to our matrix problem is $\Theta = (S^tS + \eta \Id)^{-1} S^t$. Since we are assuming $S$ symmetric, we get
\[
 \Theta = (S^2 + \eta \Id)^{-1} S.
\]

The matrix $S$ is diagonalizable with all of its eigenvalues real. The eigenvalues of $S$ are transformed by the same formula defining $\Theta$. 
If $\lambda_k$, $1 \le k \le 20L$, are the eigenvalues of $S$ then the eigenvalues of $\Theta$ will be
\[
\gamma_k = f(\lambda_k) = \frac{\lambda_k}{\lambda_k^2 + \eta}
\]
This function is well defined for all $\lambda \in \R$ when $\eta$ is positive, which proves that the matrix $S^2 + \eta \Id$ is invertible. The resulting matrix $\Theta$ is symmetric by standard matrix theory.
Finally, $f$ preserves the sign of the eigenvalue and hence $\Theta$ will be a semidefinite positive matrix whenever $S$ is.
\end{proof}

\begin{remark}
 Note that $\Theta$ can be computed by solving the linear system $(S^2 + \eta \Id) \Theta = S$, which is faster and more 
accurate than inverting the matrix $S^2 + \eta \Id$.
\end{remark}

For a better understanding of our regularization formula, we study the function $f$ in more detail. The derivative of $f$ is 
$f'(\lambda) = \frac{-\lambda^2 + \eta}{(\lambda^2+\eta)^2}$.
Hence $f$ is increasing for $|\lambda| < \sqrt{\eta}$ and decreasing for $|\lambda| > \sqrt{\eta}$, with maximum value at $\lambda = \sqrt{\eta}$ and minimum value at $\lambda = -\sqrt{\eta}$.
We show in Figure \eqref{fig:functionPlot} the plot of this function for $\eta = \eta'/\Meff = 1000/3912$ (see Section \ref{subsec:regularization} for the choice of $\eta$).

As mentioned in the proof of Proposition \ref{prop:regularization}, the function is smooth at $0$, so using this regularization formula we deal in a simple way with the conditioning problem of inverting the covariance matrix.

\begin{figure*}[!b]
\centering
\subfloat[Function $f(\lambda) = \frac{\lambda}{\lambda^2 + \eta}$, $\eta = \eta'/\Meff = 1000/3912$]{
  \includegraphics[width=0.45\textwidth]{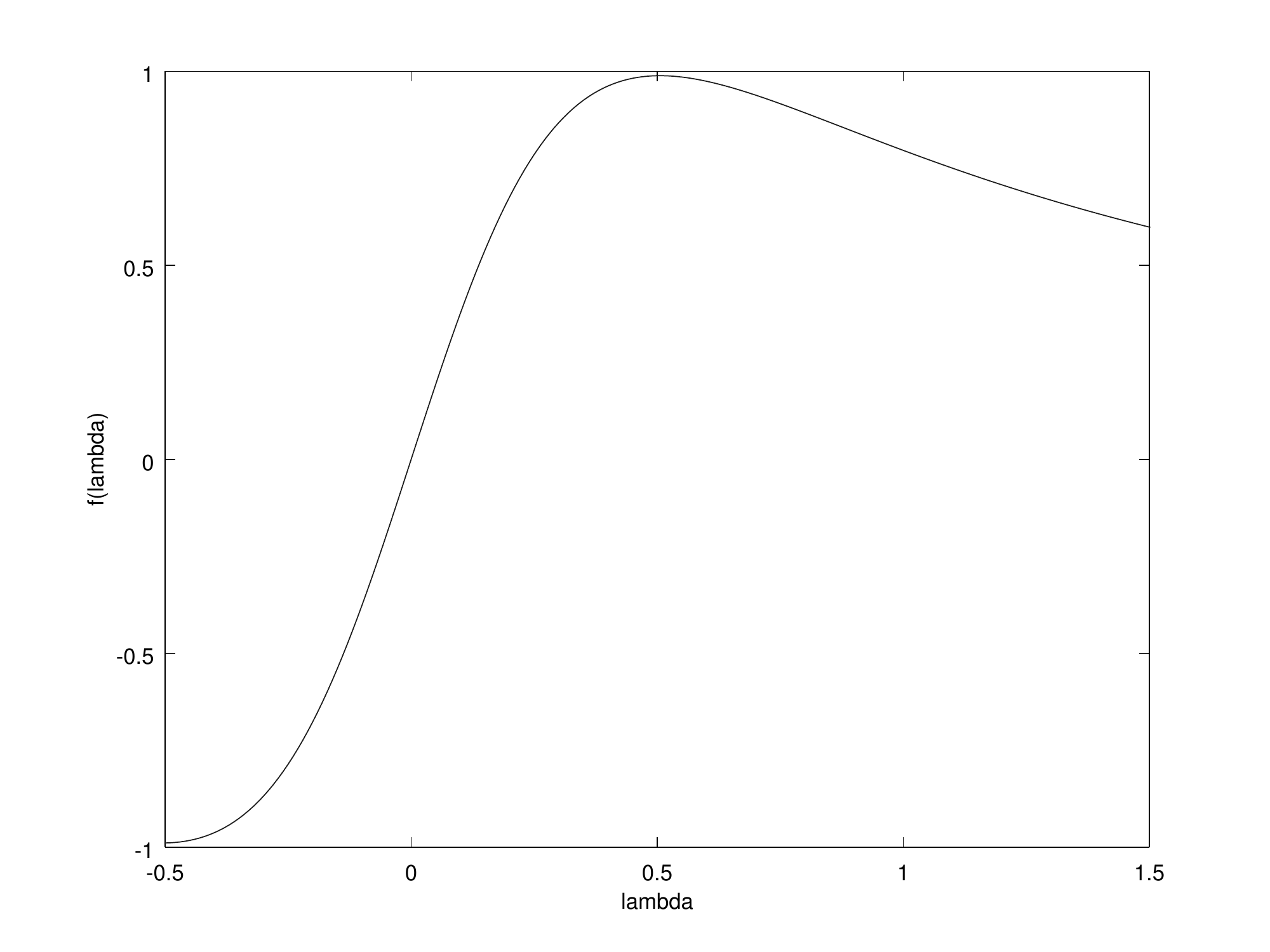}
  \label{fig:functionPlot}
}
\quad
\subfloat[Distribution of eigenvalues of the covariance matrix $S_0$]{
\includegraphics[width=0.45\textwidth]{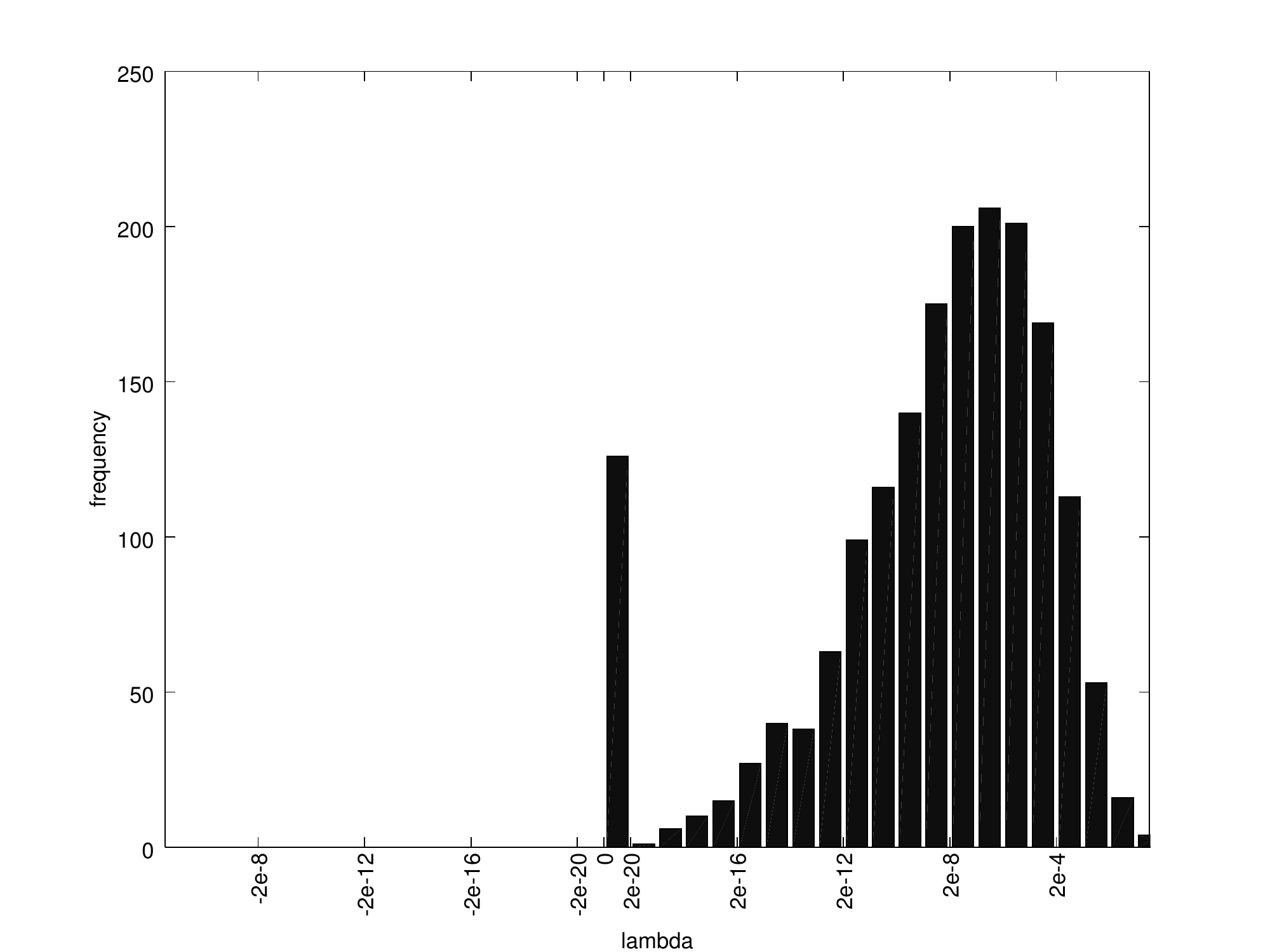}
\label{fig:eig_0_00028}
}
\hspace{0mm}
\subfloat[Relation between the eigenvalues of the modified covariance matrix $S$ and its regularized inverse $\Theta$]{
\includegraphics[width=0.45\textwidth]{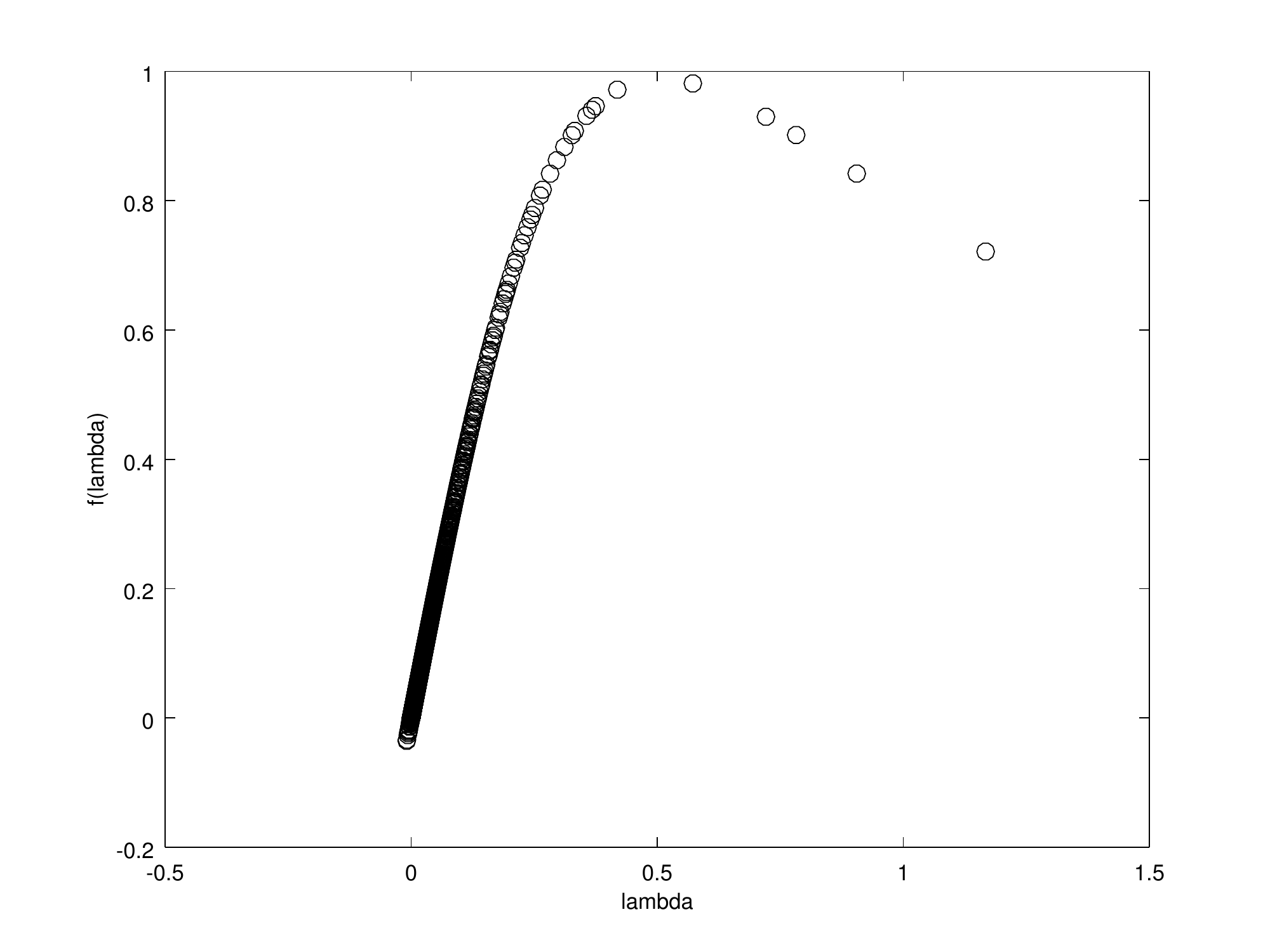}
\label{fig:transf_2_00028}
}
\quad
\subfloat[Distribution of eigenvalues of the modified covariance matrix $S$]{
\includegraphics[width=0.45\textwidth]{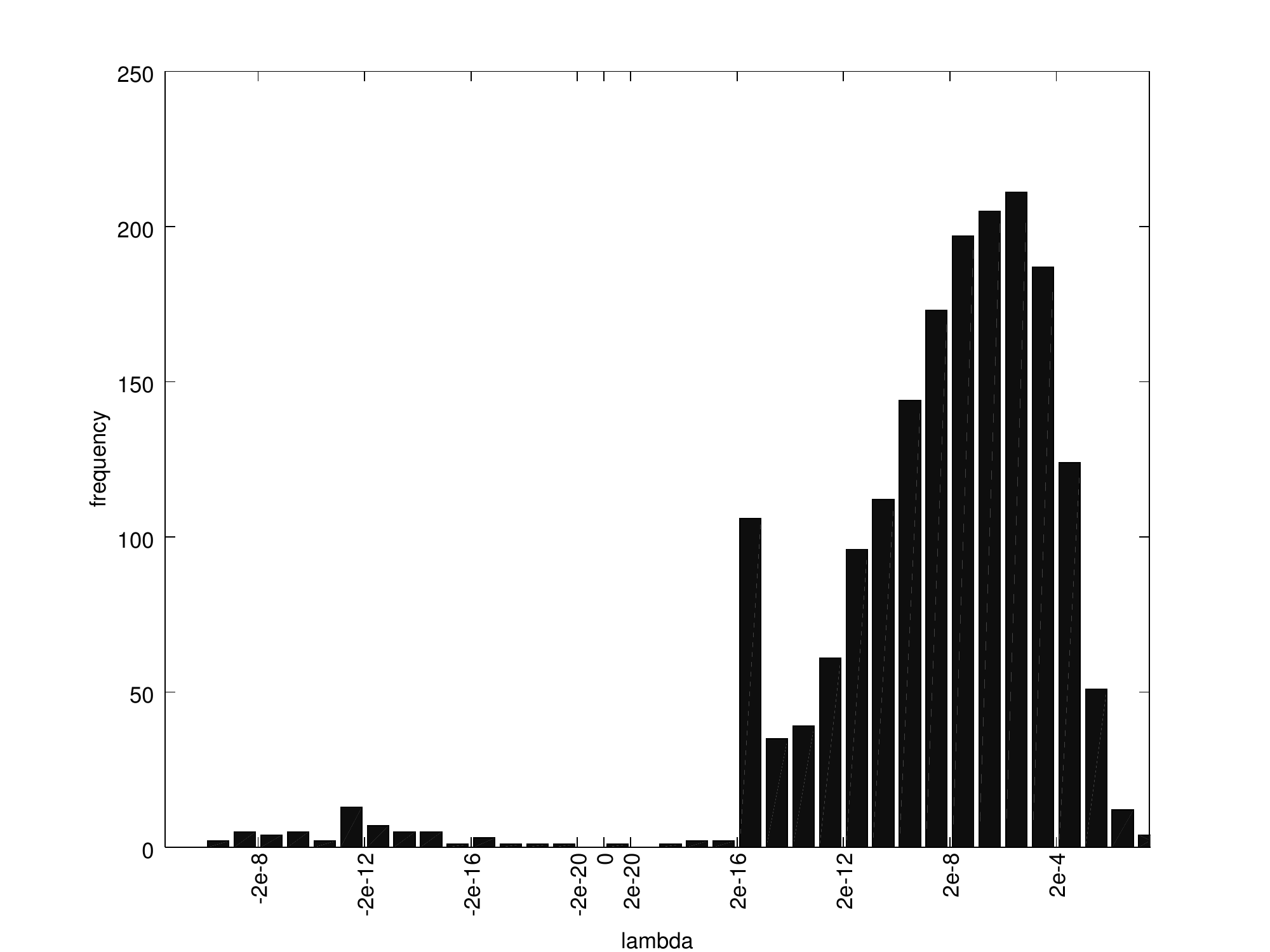}
\label{fig:eig_2_00028}
}
\caption{Distribution of eigenvalues of the covariance matrix and its regularized inverse for PFAM family PF00028.}
\end{figure*}

\subsection{Aggregation} The matrix $\Theta$ obtained is a $20L \times 20L$ matrix. Its entries are estimates of the partial correlation between pairs of columns $(i,j)$ for {\it all} pairs of amino acids $(a,b)$. Since our goal is to detect relations between pairs of columns in the alignment, we compute a coupling score aggregating the values of $\Theta$ using the $l_1$-norm on the $20 \times 20$ sub-matrices, as in \cite{jones2012psicov}. That is, 
\[
 P(i, j) = \sum_{1 \le a, b \le 20} | \Theta_{ij}(a, b) |.
\]

The $l_2$-norm for aggregation showed poorer performance than the $l_1$-norm above. Finally, following\cite{dunn2008mutual} and \cite{jones2012psicov} we define a corrected score $P_{\text{APC}}(i, j) = P(i, j) - \frac{P(\cdot, j)P(i, \cdot)}{P(\cdot, \cdot)}$, where $\cdot$ stands for the average over all positions.

The prediction of contacts between pairs of residues can now be obtained by ranking the $P_{\text{APC}}(i, j)$, where higher scores identify more likely residue-residue contacts. 

%%%%%%%%%%METHODS %%%%%%%%%%%%%%%%%%
%\begin{methods}
\section{Method details}

In this section we give more details on the actual implementation of the algorithm described above. 

\subsection{Sequence weighting}
\label{section:measure}

Families from the Pfam database contain some degree of redundancy. A common strategy to overcome this problem is sequence weighting, which weighs down groups of similar sequences and assigns higher weights to isolated sequences.

We first define a similarity measure between proteins, following \cite{smale2013introduction}. We start from the BLOSUM90 frequency substitution matrix $B_{90}(a,b)$ defined in \cite{henikoff1992amino} and call $\hat{B}_{90}(a,b)$ for a pair of amino acids $(a,b)$ the normalized matrix
\[
      \hat{B}_{90}(a,b) = \frac{B_{90}(a,b)}{\sqrt{B_{90}(a,a) B_{90}(b,b) }}
\]

We then proceed to construct a similarity kernel between pairs of proteins
\[
 K^3(p, q) = \sum_{k=1}^{10} \left(\sum_{i=1}^{L-k+1} K^2\left((p_i\dots p_{i+k-1}), (q_i\dots q_{i+k-1})\right)\right)
\]
where
\[
 K^2\left((p_i\dots p_{i+k-1}), (q_i\dots q_{i+k-1})\right) = \prod_{j=1}^k  \hat{B}_{90}(p_{i+j-1}, q_{i+j-1}),
\]
for $p, q \in \PP$, $1 \le k \le L$ and $1 \le i \le L-k+1$;

The normalized version of $K^3$ is obtained using
\[
 \hat K^3(p, q) = \frac{K^3(p, q)}{\sqrt{K^3(p, p)K^3(q,q)}}.
\]

Note that, since Pfam families consist of pre-aligned sequences, our $K^3$ kernel definition differs slightly from \cite{smale2013introduction} as it only compares aligned amino acid $k$-mers. Also, we limit the $k$-mers considered in the construction of $K^3$ to lengths smaller or equal to $10$. This implies a substantial improvement in computation time, with no significant loss in predictive power.

We fix a threshold $\theta$ (in this paper, $\theta = 0.7$) and for any protein $p \in \PP$ we define the equilibrium measure 
\[
\pi(p) = \sum_{\substack{q \in \PP \\ \hat K^3(p, q) > (1-\theta)}} \hat K^3(p, q).
\]

The weight of a protein $p$ is then defined as the reciprocal of the equilibrium measure $w(p) = (\pi(p))^{-1}$ , and the effective number of sequences in the alignment is $M_{\text{eff}} = \sum_p w(p)$. The kernel $ \hat K^3(p, q)$ is a measure of similarity between pairs of proteins, and for a given protein $p$ the quantity $\pi(p)$ effectively counts the number of sequences with similarity larger than a threshold $1 -\theta$, thereby weighing down sequences that are over-represented in the data set.

\paragraph{Hamming distance weighting}
As a term of comparison, we also applied a more traditional sequence weighting scheme based on the hamming distance between pairs of sequences. For each sequence $p$, we count the number of other sequences in the alignment that share more than $\theta$\% sequence identity with $p$
$$
	m^p = | \{ b \in \{ 1, \dots, M \} : \text{\%id}(p,b) > \theta \} |
$$
and then assign a weight $w(p) = 1 / m^p$ to sequence $p$.

This approach was used previously by several authors such as \cite{morcos2011direct}, \cite{jones2012psicov} and \cite{ekeberg2013improved}, but we note that these authors use different values for the theshold $\theta$ (respectively 0.8, 0.62 and 0.9). In this paper we choose $\theta=0.62$. See See Tables S2-S3 for the optimization of the two weighting schemes with respect to the regularization parameter $\eta'$, and Table S4 for a comparison of their performance.

\subsection{Regularization parameter}
\label{subsec:regularization}
The matrix inversion in Equation \ref{eq:theta} contains a regularization parameter $\eta$. We observed that families containing few sequences, where the number of sequences $M$ is comparable in size to the number of random variables ($20L$) require a larger regularization parameter compared to bigger families ($M \gg 20L$). We use then a regularization parameter of the form $\eta = \eta' / \Meff$, where $\Meff$ is the effective number of sequences defined above.

We tried different values of $\eta'$ over the 15 families from \cite{marks2011protein}, and observed that in general roughly the same $\eta'$ appears to be optimal across families with different $\Meff$. Thus the normalization $\eta = \eta' / \Meff$ appears appropriate.

In Figure \eqref{fig:transf_2_00028} we show how the actual eigenvalues of the modified covariance matrix corresponding to PFAM family PF00028 are transformed when computing the regularized inverse.

\subsection{Pseudocounts}
\label{subsec:pseudocounts}
The pseudocounts parameter $\lambda$ in Equation \ref{eq:pseudo} accounts for non-observed pairs of amino acids. Following \cite{morcos2011direct}, we set $\lambda=44$. However, we observed that the performance gain is small compared to setting $\lambda=0$ (see Table S5). In fact, pseudocounts have a similar regularizing effect to the parameter $\eta'$ described in the previous section, and probably for this reason the contribution of $\lambda$ is minimal. 

%\end{methods}

%%%%%%%%%%%%%%%%%%%%%%%%%%%%%%%%%%%%%%%%%%%%%%%%%%%%%%%%%%%%%%%%%%%%%%%%%%%%%%%%%%%%%
%
%     please remove the " % " symbol from \centerline{\includegraphics{fig01.eps}}
%     as it may ignore the figures.
%
%%%%%%%%%%%%%%%%%%%%%%%%%%%%%%%%%%%%%%%%%%%%%%%%%%%%%%%%%%%%%%%%%%%%%%%%%%%%%%%%%%%%%%

%%% TABLES %%%%%%%%%%%%%%%%%%%%%

%%%%%%% 15 families - C-Beta

\begin{table*}[!b]
%\processtable{Positive predictive value for partial correlation scores in the set of 15 families from \cite{marks2011protein}, measured on C$\beta$ carbons \label{table:results15}}
\caption{Positive predictive value for partial correlation scores in the set of 15 families from \cite{marks2011protein}, measured on C$\beta$ carbons \label{table:results15}}
%{
\centering
{\tiny
\begin{tabular}{p{0.7cm}p{0.3cm}p{0.2cm}p{0.3cm}c|ccc|ccc|ccc|ccc}
\toprule
\multirow{2}*{PFAM} & \multirow{2}*{PDB} & \multirow{2}*{L} & \multirow{2}*{M} &  \multirow{2}*{$\Meff$} &
\multicolumn{3}{c|}{top $L/5$} & \multicolumn{3}{c|}{top $L/3$} & \multicolumn{3}{c|}{top $L/2$} & \multicolumn{3}{c}{top $L$} \\ 
&&&&& RLS & PSICOV & DCA & RLS & PSICOV & DCA & RLS & PSICOV & DCA & RLS & PSICOV & DCA \\ \midrule

PF00001 & 1hzx & $257$ & 23711 & 3610         & 0.608 & 0.549 & 0.417 & 0.565 & 0.494 & 0.333 & 0.516 & 0.453 & 0.281 & 0.374 & 0.300 & 0.222 \\
PF00013 & 1wvn & $56$ & 5298 & 1710            & 0.909 & 0.818 & 0.909 & 0.889 & 0.778 & 0.944 & 0.857 & 0.643 & 0.926 & 0.679 & 0.518 & 0.655 \\
PF00014 & 5pti & $48$ & 1743 & 1019     & 0.889 & 0.889 & 1.000 & 0.625 & 0.750 & 0.750 & 0.750 & 0.750 & 0.750 & 0.625 & 0.542 & 0.604 \\
PF00018 & 2hda & $45$ & 3610 & 1529           & 0.889 & 0.889 & 0.889 & 0.933 & 0.933 & 0.933 & 0.909 & 0.864 & 0.818 & 0.667 & 0.622 & 0.667 \\
PF00028 & 2o72 & $91$ & 8828 & 3912         & 0.944 & 0.944 & 0.823 & 0.933 & 0.933 & 0.862 & 0.889 & 0.889 & 0.909 & 0.802 & 0.791 & 0.809 \\
PF00059 & 2it6 & $107$ & 4067 & 1949       & 0.810 & 0.810 & 0.850 & 0.771 & 0.743 & 0.765 & 0.774 & 0.736 & 0.726 & 0.598 & 0.561 & 0.628 \\
PF00071 & 5p21 & $161$ & 8395 & 1868             & 0.844 & 0.813 & 0.774 & 0.811 & 0.755 & 0.811 & 0.763 & 0.725 & 0.798 & 0.652 & 0.627 & 0.604 \\
PF00072 & 1e6k & $108$ & 45821 & 24642 & 0.810 & 0.857 & 0.857 & 0.778 & 0.861 & 0.857 & 0.759 & 0.815 & 0.793 & 0.685 & 0.694 & 0.755 \\
PF00075 & 1f21 & $126$ & 8131 & 960         & 0.760 & 0.680 & 0.680 & 0.738 & 0.714 & 0.610 & 0.714 & 0.667 & 0.581 & 0.579 & 0.524 & 0.520 \\
PF00076 & 1g2e & $70$ & 18491 & 6849          & 0.786 & 0.786 & 0.786 & 0.870 & 0.783 & 0.870 & 0.857 & 0.829 & 0.914 & 0.757 & 0.743 & 0.814 \\
PF00085 & 1rqm & $100$ & 9095 & 3814     & 0.800 & 0.700 & 0.842 & 0.758 & 0.758 & 0.688 & 0.740 & 0.700 & 0.729 & 0.650 & 0.570 & 0.639 \\
PF00089 & 3tgi & $217$ & 12909 & 4296        & 0.977 & 0.837 & 0.951 & 0.847 & 0.750 & 0.855 & 0.796 & 0.713 & 0.789 & 0.724 & 0.595 & 0.673 \\
PF00254 & 1r9h & $95$ & 5269 & 1759          & 0.895 & 0.895 & 0.889 & 0.807 & 0.839 & 0.936 & 0.851 & 0.745 & 0.891 & 0.674 & 0.579 & 0.667 \\
PF00307 & 1bkr & $107$ & 2751 & 873              & 0.714 & 0.333 & 0.700 & 0.629 & 0.229 & 0.697 & 0.547 & 0.264 & 0.500 & 0.449 & 0.252 & 0.455 \\
PF00486 & 1odd & $74$ & 13702 & 5452  & 0.786 & 0.786 & 0.714 & 0.708 & 0.583 & 0.667 & 0.649 & 0.541 & 0.583 & 0.527 & 0.419 & 0.521 \\ \midrule

& & & \multicolumn{2}{c}{Average} & 
\bf 0.828 & 0.772 & 0.805 & \bf 0.777 & 0.727 & 0.772 & \bf 0.758 & 0.689 & 0.732 & \bf 0.629 & 0.556 & 0.615 \\ 

%\multicolumn{5}{c}{Best in} & 
%$\bf 10 / 15$ & $8 / 15$ & $\bf 10/15$ & $8 / 15$ & $\bf 12/15$ & $5 / 15$ & $\bf 10/15$ & $ 6/15$ \\ \botrule
\end{tabular}
}
%}{}
\end{table*}

%%%%%%% 22 families - Ekeberg set %%%%%%

\begin{table*}[!b]
%\processtable{Positive predictive value for partial correlation scores in a set of 22 families from \cite{ekeberg2013improved}, measured on C$\beta$ carbons \label{table:results22}}
\caption{Positive predictive value for partial correlation scores in a set of 22 families from \cite{ekeberg2013improved}, measured on C$\beta$ carbons \label{table:results22}}
%{
\centering
{\tiny
\begin{tabular}{p{0.7cm}p{0.3cm}p{0.2cm}p{0.3cm}c|ccc|ccc|ccc|ccc}
\toprule
\multirow{2}*{PFAM} & \multirow{2}*{PDB} & \multirow{2}*{L} & \multirow{2}*{M} &  \multirow{2}*{$\Meff$} &
\multicolumn{3}{c|}{top $L/5$} & \multicolumn{3}{c|}{top $L/3$} & \multicolumn{3}{c|}{top $L/2$} & \multicolumn{3}{c}{top $L$} \\ 
&&&&& RLS & PSICOV & DCA & RLS & PSICOV & DCA & RLS & PSICOV & DCA & RLS & PSICOV & DCA \\ \midrule

PF00006 & 2r9v & 213 & 13515 & 161  & 0.429 & 0.310 & 0.634 & 0.366 & 0.282 & 0.638 & 0.330 & 0.255 & 0.625 & 0.239 & 0.202 & 0.450 \\
PF00011 & 2bol & 100 & 6600 & 2622  & 0.800 & 0.750 & 0.875 & 0.788 & 0.576 & 0.815 & 0.660 & 0.520 & 0.756 & 0.530 & 0.450 & 0.578 \\
PF00017 & 1o47 & 75 & 4911 & 1361   & 0.800 & 0.733 & 0.733 & 0.760 & 0.680 & 0.680 & 0.676 & 0.568 & 0.622 & 0.493 & 0.440 & 0.547 \\
PF00025 & 1fzq & 172 & 4384 & 659   & 0.647 & 0.559 & 0.500 & 0.544 & 0.509 & 0.518 & 0.454 & 0.407 & 0.482 & 0.326 & 0.308 & 0.324 \\
PF00026 & 3er5 & 317 & 4326 & 1445  & 0.810 & 0.762 & 0.800 & 0.752 & 0.667 & 0.733 & 0.753 & 0.658 & 0.717 & 0.596 & 0.489 & 0.586 \\
PF00027 & 3fhi & 89 & 16544 & 8545  & 1.000 & 1.000 & 1.000 & 1.000 & 0.966 & 1.000 & 0.977 & 0.909 & 1.000 & 0.843 & 0.809 & 0.888 \\
PF00032 & 1zrt & 80 & 30643 & 702   & 0.750 & 0.688 & 0.500 & 0.615 & 0.615 & 0.462 & 0.500 & 0.450 & 0.400 & 0.263 & 0.263 & 0.238 \\
PF00035 & 1o0w & 64 & 4254 & 1958   & 0.917 & 0.917 & 1.000 & 0.952 & 0.857 & 0.952 & 0.906 & 0.688 & 0.903 & 0.609 & 0.531 & 0.714 \\
PF00041 & 1bqu & 82 & 33103 & 11568 & 0.688 & 0.938 & 0.813 & 0.741 & 0.704 & 0.731 & 0.659 & 0.610 & 0.650 & 0.585 & 0.524 & 0.588 \\
PF00043 & 6gsu & 93 & 8822 & 2653   & 0.556 & 0.500 & 0.563 & 0.452 & 0.419 & 0.500 & 0.348 & 0.370 & 0.429 & 0.258 & 0.237 & 0.321 \\
PF00044 & 1crw & 144 & 8047 & 838   & 0.893 & 0.857 & 0.929 & 0.813 & 0.833 & 0.870 & 0.722 & 0.736 & 0.814 & 0.618 & 0.611 & 0.700 \\
PF00046 & 2vi6 & 57 & 10117 & 1376  & 0.455 & 0.636 & 0.500 & 0.526 & 0.474 & 0.471 & 0.429 & 0.393 & 0.385 & 0.281 & 0.263 & 0.264 \\
PF00056 & 1a5z & 137 & 5673 & 813   & 0.852 & 0.704 & 0.926 & 0.844 & 0.578 & 0.867 & 0.765 & 0.529 & 0.821 & 0.569 & 0.416 & 0.615 \\
PF00073 & 2r06 & 170 & 11796 & 136  & 0.353 & 0.206 & 0.094 & 0.250 & 0.125 & 0.111 & 0.224 & 0.094 & 0.111 & 0.165 & 0.112 & 0.117 \\
PF00081 & 3bfr & 76 & 4017 & 619    & 0.600 & 0.467 & 0.429 & 0.520 & 0.480 & 0.458 & 0.421 & 0.368 & 0.514 & 0.263 & 0.290 & 0.351 \\
PF00084 & 1elv & 53 & 13038 & 5457  & 0.700 & 0.600 & 0.667 & 0.706 & 0.471 & 0.625 & 0.692 & 0.577 & 0.625 & 0.472 & 0.491 & 0.521 \\
PF00091 & 2r75 & 215 & 10854 & 321  & 0.442 & 0.488 & 0.706 & 0.352 & 0.380 & 0.614 & 0.299 & 0.355 & 0.547 & 0.265 & 0.261 & 0.407 \\
PF00092 & 1atz & 177 & 7070 & 1799  & 0.714 & 0.771 & 0.758 & 0.695 & 0.780 & 0.582 & 0.716 & 0.727 & 0.566 & 0.588 & 0.593 & 0.530 \\
PF00105 & 1gdc & 64 & 3595 & 663    & 0.667 & 0.667 & 0.833 & 0.619 & 0.619 & 0.619 & 0.500 & 0.500 & 0.581 & 0.344 & 0.375 & 0.444 \\
PF00107 & 1a71 & 128 & 24755 & 7021 & 0.880 & 0.840 & 0.913 & 0.833 & 0.810 & 0.790 & 0.750 & 0.781 & 0.776 & 0.633 & 0.602 & 0.690 \\
PF00108 & 3goa & 261 & 9701 & 1680  & 0.750 & 0.731 & 0.778 & 0.713 & 0.644 & 0.733 & 0.685 & 0.685 & 0.735 & 0.605 & 0.563 & 0.678 \\
PF00111 & 1a70 & 74 & 10559 & 2848  & 0.643 & 0.571 & 0.615 & 0.542 & 0.458 & 0.409 & 0.432 & 0.405 & 0.394 & 0.405 & 0.378 & 0.303 \\   \midrule
& & & \multicolumn{2}{c}{Average} & 
0.697 & 0.668 & \bf 0.708 &\bf 0.654 & 0.588 & 0.644 & 0.586 & 0.527 & \bf 0.612 & 0.452 & 0.418 & \bf 0.493\\ 

%\multicolumn{5}{c}{Best in} & 
%$\bf 8 / 10$ & $3 / 10$ & $\bf 9/10$ & $2 / 10$ & $\bf 9/10$ & $2 / 10$ & $\bf 8/10$ & $ 4/10$ \\ \botrule
\end{tabular}
}
%}{}
\end{table*}

%%%%%%% 10 families - Cbeta %%%%%%

\begin{table*}[!b]
%\processtable{Positive predictive value for partial correlation scores in a new set of 10 families, measured on C$\beta$ carbons \label{table:results10}}
\caption{Positive predictive value for partial correlation scores in a new set of 10 families, measured on C$\beta$ carbons \label{table:results10}}
%{
\centering
{\tiny
\begin{tabular}{p{0.7cm}p{0.3cm}p{0.2cm}p{0.3cm}c|ccc|ccc|ccc|ccc}
\toprule
\multirow{2}*{PFAM} & \multirow{2}*{PDB} & \multirow{2}*{L} & \multirow{2}*{M} &  \multirow{2}*{$\Meff$} &
\multicolumn{3}{c|}{top $L/5$} & \multicolumn{3}{c|}{top $L/3$} & \multicolumn{3}{c|}{top $L/2$} & \multicolumn{3}{c}{top $L$} \\ 
&&&&& RLS & PSICOV & DCA & RLS & PSICOV & DCA & RLS & PSICOV & DCA & RLS & PSICOV & DCA \\ \midrule

PF00240 & 4k7s & $69$  & 5382 & 1809         & 0.923 & 0.846 & 0.846 & 0.870 & 0.783 & 0.905 & 0.765 & 0.677 & 0.875 & 0.580 & 0.449 & 0.600 \\
PF00390 &  1qr6 & $182$ & 3865 & 182          & 0.583 & 0.583 & 0.314 & 0.467 & 0.450 & 0.373 & 0.396 & 0.330 & 0.292 & 0.247 & 0.242 & 0.225 \\
PF00482 &  2vmb & $124$ & 8131 & 3673      & 0.833 & 0.583 & 0.857 & 0.829 & 0.537 & 0.800 & 0.742 & 0.403 & 0.808 & 0.460 & 0.282 & 0.590 \\
PF00793 &  3stc & $256$ & 4878 & 506          & 0.628 & 0.294 & 0.556 & 0.494 & 0.294 & 0.513 & 0.414 & 0.234 & 0.474 & 0.289 & 0.184 & 0.411 \\
PF01026 &  1xwy & $244$ & 5507 & 1907     & 0.833 & 0.771 & 0.792 & 0.815 & 0.741 & 0.775 & 0.746 & 0.664 & 0.758 & 0.648 & 0.562 & 0.679 \\
PF01791 &  3myp & $233$ & 3865 & 673       & 0.261 & 0.304 & 0.273 & 0.286 & 0.286 & 0.288 & 0.250 & 0.267 & 0.255 & 0.206 & 0.197 & 0.223 \\
PF01807 &  2au3 & $92$  & 2773 & 837        & 0.778 & 0.833 & 0.824 & 0.800 & 0.767 & 0.862 & 0.826 & 0.848 & 0.886 & 0.707 & 0.707 & 0.753 \\
PF03484 &  1b7y & $62$  & 2669 & 1265       & 1.000 & 0.750 & 1.000 & 0.950 & 0.850 & 0.900 & 0.903 & 0.774 & 0.933 & 0.726 & 0.629 & 0.754 \\
PF12704 &  3ftj & $231$ & 14338 & 7373      & 0.891 & 0.717 & 0.919 & 0.844 & 0.546 & 0.853 & 0.800 & 0.452 & 0.804 & 0.597 & 0.286 & 0.643 \\
PF13305 &  3on2 & $80$  & 1583 & 964         & 0.500 & 0.438 & 0.400 & 0.308 & 0.269 & 0.360 & 0.225 & 0.200 & 0.263 & 0.163 & 0.163 & 0.145 \\ \midrule
& & & \multicolumn{2}{c}{Average} & 
\bf 0.723 & 0.612 & 0.678 & \bf 0.666 & 0.552 & 0.663 & 0.607 & 0.485 & \bf 0.635 & 0.462 & 0.370 & \bf 0.502 \\ 
%\multicolumn{5}{c}{Best in} & 
%$\bf 8 / 10$ & $3 / 10$ & $\bf 9/10$ & $2 / 10$ & $\bf 9/10$ & $2 / 10$ & $\bf 8/10$ & $ 4/10$ \\ \botrule
\end{tabular}
}
%}{}
\end{table*}

\section{Results and Conclusion}

The method and estimation of parameters described above were first applied to the 15 families studied in Marks et al.\ \cite{marks2011protein}. Performance was estimated in terms of the fraction of correct predicted contacts among the $L/5$, $L/3$, $L/2$ and $L$ pairs with highest $P_{\text{APC}}$ score, where $L$ is the length of the alignment. We considered as a true contact a pair of amino acids with beta-carbons (C$\beta$) with distance $< 8$ {\AA} and at least 5 residues apart along the length of the protein. We find that on these 15 families the optimal value for the regularization parameter is around $\eta'=1000$ (see Table S2). 

Table \ref{table:results15} compares the performance of the RLS algorithm with PSICOV version 1.11 \cite{jones2012psicov} and the plmDCA method \cite{ekeberg2013improved}. We observe that on this set our method outperforms both PSICOV and DCA on the majority of families. Additionally, in Table S6 in shown the positive predictive value of the methods with respect to short range ($5 \leq i - j \leq 11$), medium range ($12 \leq i - j \leq 23$) and long range ($>$ 23) interactions for the prediction of the top $L/5$ contacts.

Next, we applied the three methods using the same parameters on an additional set of families from \cite{ekeberg2013improved}. This set partially overlapped with the families from  \cite{marks2011protein} studied in Table \ref{table:results15}, and after removing the duplicates we are left with a set of 22 families. Table \ref{table:results22} shows the positive predictive value of the RLS algorithm compared to PSICOV and plmDCA. On this set plmDCA obtains the highest average performance on three out of four ranking categories. However, note that plmDCA was optimized on this set of families, so there may be a bias in favor of this method.

Finally, we constructed an independent set of 10 families, selected randomly from Pfam with the only condition of containing at least 1,000 unique sequences. This set had not been used in the optimization of the algorithms, therefore constitutes a fair ground for comparison. The results (see Table \ref{table:results10}) show that RLS outperforms the other two methods in the $L/5$ and $L/3$ subsets, whereas DCA obtains highest average performance for the prediction of the top $L/2$ and $L$ contacts.

A comparison of the running times of the two best algorithms (RLS and DCA) shows that RLS is at least one order of magnitude faster than DCA (Table \ref{table:timings}). Note that we used the latest fast version of plmDCA \cite{ekeberg2014fast}, termed ``asymmetric plmDCA'', which improves considerably on previous pseudolikelihood methods in terms of speed. Our fast regularized inversion of the covariance matrix allows contact prediction on hundreds of amino acids-long domains in a matter of seconds, practically removing the limitations on the length of proteins that can be analyzed. In fact, the slowest step in the predictions is sequence weighting (not accounted in Table \ref{table:timings} for either method), in particular the $K^3$-based weighting can be slow for very large families, but we showed that a simpler and faster weighting strategy does not affect too dramatically the performance (Table S4).

\begin{table}[!t]
%\processtable{Running times for the RLS and plmDCA (asymmetric) algorithms, in seconds \label{table:timings}. Both algorithms were run on the same computer, using a fixed number of 4 processors.}
\caption{Running times for the RLS and plmDCA (asymmetric) algorithms, in seconds. Both algorithms were run on the same computer, using a fixed number of 4 processors. \label{table:timings}}
%{
\centering
{\scriptsize
\begin{tabular}{c|c|c|c}
Pfam & Length & RLS & plmDCA \\
\hline
PF00001 & 257 & 21.9 & 2429.5 \\
PF00013 & 56 & 0.3 & 25.8 \\
PF00014 & 48 & 0.2 & 9.1 \\
PF00018 & 45 & 0.2 & 11.9 \\
PF00028 & 91 & 1.2 & 103.5 \\
PF00059 & 107 & 1.5 & 73.2 \\
PF00071 & 161 & 5.0 & 319.0 \\
PF00072 & 108 & 5.0 & 686.0 \\
PF00075 & 126 & 2.7 & 177.9 \\
PF00076 & 70 & 1.0 & 108.1 \\
PF00085 & 100 & 1.6 & 128.7 \\
PF00089 & 217 & 11.5 & 1029.3 \\
PF00254 & 95 & 1.2 & 69.8 \\
PF00307 & 107 & 1.4 & 49.0 \\
PF00486 & 74 & 0.9 & 104.6 \\
\end{tabular}
}
%}{}
\end{table}

In general, we observed that the performance depends on the effective number of sequences $\Meff$ in the alignment. For instance, families PF00390 or PF00793 are composed of several thousand sequences, but they contain much redundancy, which brings down $\Meff$ to a few hundred units. Roughly, it appears that at least 1000 non-redundant sequences ($\Meff > 1000$) are necessary to achieve a reasonable precision for contact prediction. This is in agreement with previous estimates \cite{marks2011protein, kamisetty2013assessing} which place this number at about $5L$, where $L$ is the length of the alignment.

In conclusion, we demonstrated how our simple regularization scheme for covariance matrix inversion allows the fast and accurate  prediction of residue-residue contacts. Currently, a major restriction to this kind of approach is the fairly high number of non-redundant sequences required to infer coevolution from a multiple sequence alignment, limiting the application to a relatively small subset of Pfam.  However, as the number of protein sequences deposited in public databases increases, we expect a larger number of protein families to become accessible to our analysis, as well as improved performance on those that are already accessible.

\section*{Acknowledgement}

This work received funding from City University of Hong Kong grants RGC \#9380050 and \#9041544. Work by S.L. was partially supported by Ministerio de Ciencia, Tecnolog\'ia e Innovaci\'on Productiva, Argentina. 

%\bibliographystyle{natbib}
%\bibliographystyle{achemnat}
%\bibliographystyle{plainnat}
%\bibliographystyle{abbrv}
%\bibliographystyle{bioinformatics}
%
%\bibliographystyle{plain}
%
%\bibliography{Document}

\bibliographystyle{plain}

\bibliography{mybib_plain}

\begin{thebibliography}{10}

\bibitem{anfinsen1973principles}
Christian~B Anfinsen.
\newblock Principles that govern the folding of protein chains.
\newblock {\em Science}, 181(4096):223--230, 1973.

\bibitem{balakrishnan2011learning}
Sivaraman Balakrishnan, Hetunandan Kamisetty, Jaime~G Carbonell, Su-In Lee, and
  Christopher~James Langmead.
\newblock Learning generative models for protein fold families.
\newblock {\em Proteins: Structure, Function, and Bioinformatics},
  79(4):1061--1078, 2011.

\bibitem{burger2010disentangling}
Lukas Burger and Erik van Nimwegen.
\newblock Disentangling direct from indirect co-evolution of residues in
  protein alignments.
\newblock {\em PLoS computational biology}, 6(1):e1000633, 2010.

\bibitem{buslje2009correction}
Cristina~Marino Buslje, Javier Santos, Jose~Maria Delfino, and Morten Nielsen.
\newblock Correction for phylogeny, small number of observations and data
  redundancy improves the identification of coevolving amino acid pairs using
  mutual information.
\newblock {\em Bioinformatics}, 25(9):1125--1131, 2009.

\bibitem{de2013emerging}
David de~Juan, Florencio Pazos, and Alfonso Valencia.
\newblock Emerging methods in protein co-evolution.
\newblock {\em Nature Reviews Genetics}, 14(4):249--261, 2013.

\bibitem{dill2012protein}
Ken~A Dill and Justin~L MacCallum.
\newblock The protein-folding problem, 50 years on.
\newblock {\em Science}, 338(6110):1042--1046, 2012.

\bibitem{dunn2008mutual}
Stanley~D Dunn, Lindi~M Wahl, and Gregory~B Gloor.
\newblock Mutual information without the influence of phylogeny or entropy
  dramatically improves residue contact prediction.
\newblock {\em Bioinformatics}, 24(3):333--340, 2008.

\bibitem{ekeberg2014fast}
Magnus Ekeberg, Tuomo Hartonen, and Erik Aurell.
\newblock Fast pseudolikelihood maximization for direct-coupling analysis of
  protein structure from many homologous amino-acid sequences.
\newblock {\em arXiv preprint arXiv:1401.4832}, 2014.

\bibitem{ekeberg2013improved}
Magnus Ekeberg, Cecilia L{\"o}vkvist, Yueheng Lan, Martin Weigt, and Erik
  Aurell.
\newblock Improved contact prediction in proteins: using pseudolikelihoods to
  infer {P}otts models.
\newblock {\em Physical Review E}, 87(1):012707, 2013.

\bibitem{henikoff1992amino}
Steven Henikoff and Jorja~G Henikoff.
\newblock Amino acid substitution matrices from protein blocks.
\newblock {\em Proceedings of the National Academy of Sciences},
  89(22):10915--10919, 1992.

\bibitem{hoerl1962application}
Arthur~E Hoerl.
\newblock Application of ridge analysis to regression problems.
\newblock {\em Chemical Engineering Progress}, 58(3):54--59, 1962.

\bibitem{jones2012psicov}
David~T Jones, Daniel~WA Buchan, Domenico Cozzetto, and Massimiliano Pontil.
\newblock {PSICOV}: precise structural contact prediction using sparse inverse
  covariance estimation on large multiple sequence alignments.
\newblock {\em Bioinformatics}, 28(2):184--190, 2012.

\bibitem{kamisetty2013assessing}
Hetunandan Kamisetty, Sergey Ovchinnikov, and David Baker.
\newblock Assessing the utility of coevolution-based residue--residue contact
  predictions in a sequence and structure-rich era.
\newblock {\em Proceedings of the National Academy of Sciences},
  110(39):15674--15679, 2013.

\bibitem{lapedes1999correlated}
Alan~S Lapedes, Bertrand~G Giraud, LonChang Liu, and Gary~D Stormo.
\newblock Correlated mutations in models of protein sequences: phylogenetic and
  structural effects.
\newblock {\em Lecture Notes-Monograph Series}, pp. 236--256, 1999.

\bibitem{marks2011protein}
Debora~S Marks, Lucy~J Colwell, Robert Sheridan, Thomas~A Hopf, Andrea Pagnani,
  Riccardo Zecchina, and Chris Sander.
\newblock Protein {3D} structure computed from evolutionary sequence variation.
\newblock {\em PLoS One}, 6(12):e28766, 2011.

\bibitem{marks2012protein}
Debora~S Marks, Thomas~A Hopf, and Chris Sander.
\newblock Protein structure prediction from sequence variation.
\newblock {\em Nature biotechnology}, 30(11):1072--1080, 2012.

\bibitem{morcos2011direct}
Faruck Morcos, Andrea Pagnani, Bryan Lunt, Arianna Bertolino, Debora~S Marks,
  Chris Sander, Riccardo Zecchina, Jos{\'e}~N Onuchic, Terence Hwa, and Martin
  Weigt.
\newblock Direct-coupling analysis of residue coevolution captures native
  contacts across many protein families.
\newblock {\em Proceedings of the National Academy of Sciences},
  108(49):E1293--E1301, 2011.

\bibitem{PFAM}
Marco Punta, Penny~C Coggill, Ruth~Y Eberhardt, Jaina Mistry, John Tate, Chris
  Boursnell, Ningze Pang, Kristoffer Forslund, Goran Ceric, Jody Clements,
  et~al.
\newblock The {P}fam protein families database.
\newblock {\em Nucleic acids research}, 40(D1):D290--D301, 2012.

\bibitem{rose2006backbone}
George~D Rose, Patrick~J Fleming, Jayanth~R Banavar, and Amos Maritan.
\newblock A backbone-based theory of protein folding.
\newblock {\em Proceedings of the National Academy of Sciences},
  103(45):16623--16633, 2006.

\bibitem{smale2013introduction}
Wen-Jun Shen, Hau-San Wong, Quan-Wu Xiao, Xin Guo, and Stephen Smale.
\newblock Introduction to the peptide binding problem of computational
  immunology: New results.
\newblock {\em Foundations of Computational Mathematics}, pp. 1--34, 2013.

\bibitem{tikhonov1943stability}
Andrey~Nikolayevich Tikhonov.
\newblock On the stability of inverse problems.
\newblock {\em Dokl. Akad. Nauk SSSR}, 39(5):195--198, 1943.

\bibitem{wang2013predicting}
Zhiyong Wang and Jinbo Xu.
\newblock Predicting protein contact map using evolutionary and physical
  constraints by integer programming.
\newblock {\em Bioinformatics}, 29(13):i266--i273, 2013.

\bibitem{weigt2009identification}
Martin Weigt, Robert~A White, Hendrik Szurmant, James~A Hoch, and Terence Hwa.
\newblock Identification of direct residue contacts in protein--protein
  interaction by message passing.
\newblock {\em Proceedings of the National Academy of Sciences}, 106(1):67--72,
  2009.

\end{thebibliography}

%%Supplementary material%%%
\clearpage
\includepdf[pages=1,landscape=false]{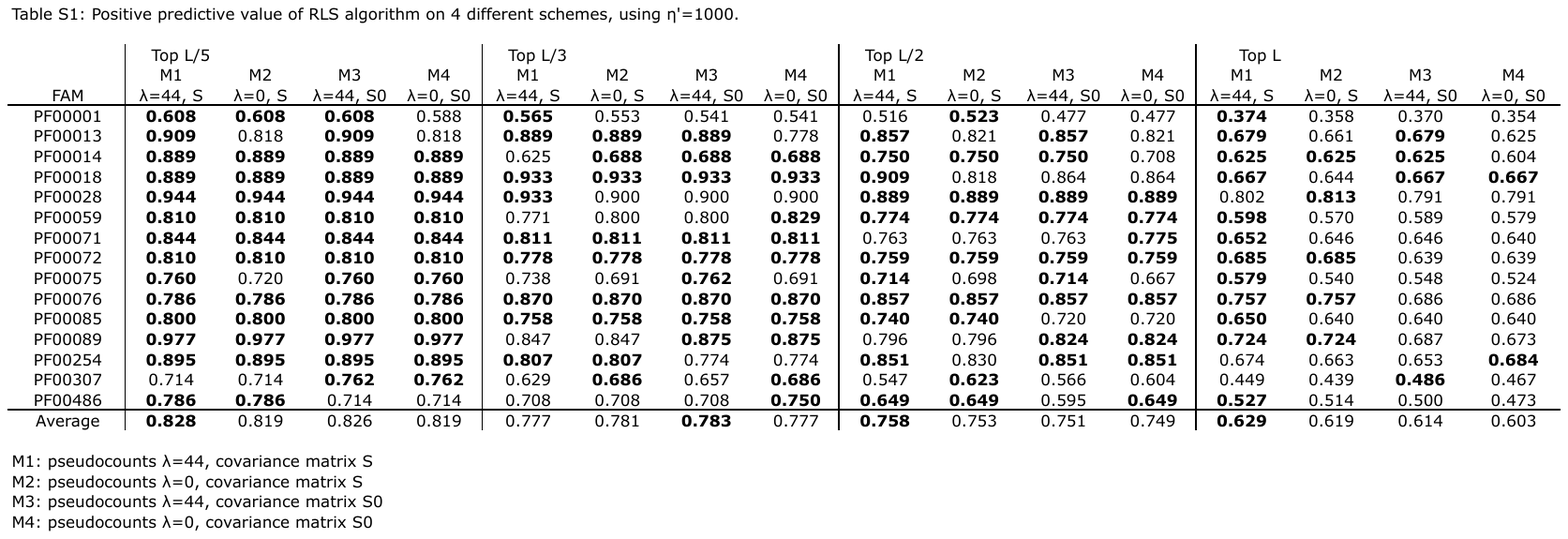}
\includepdf[pages=1,landscape=false]{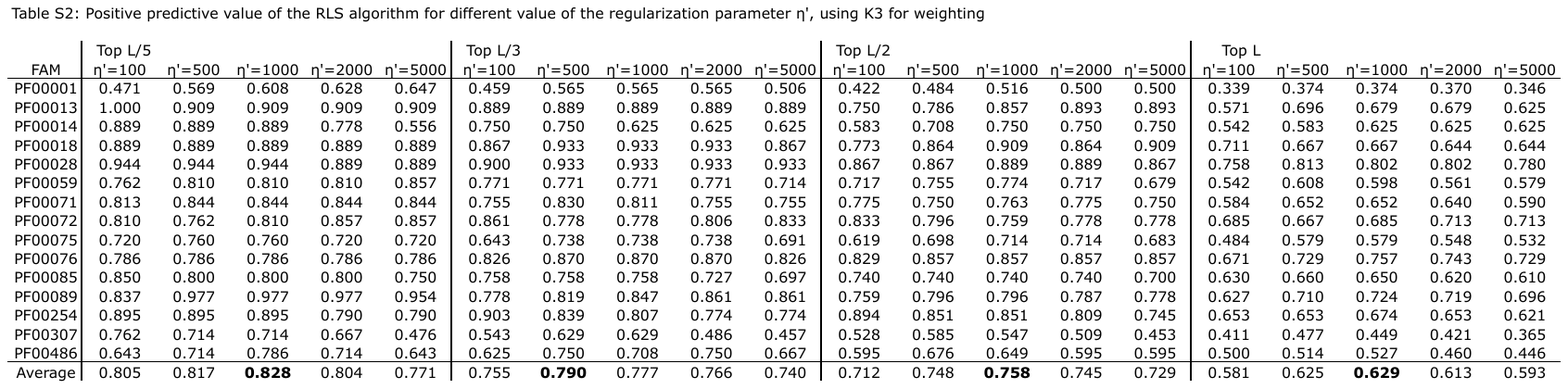}
\includepdf[pages=1,landscape=false]{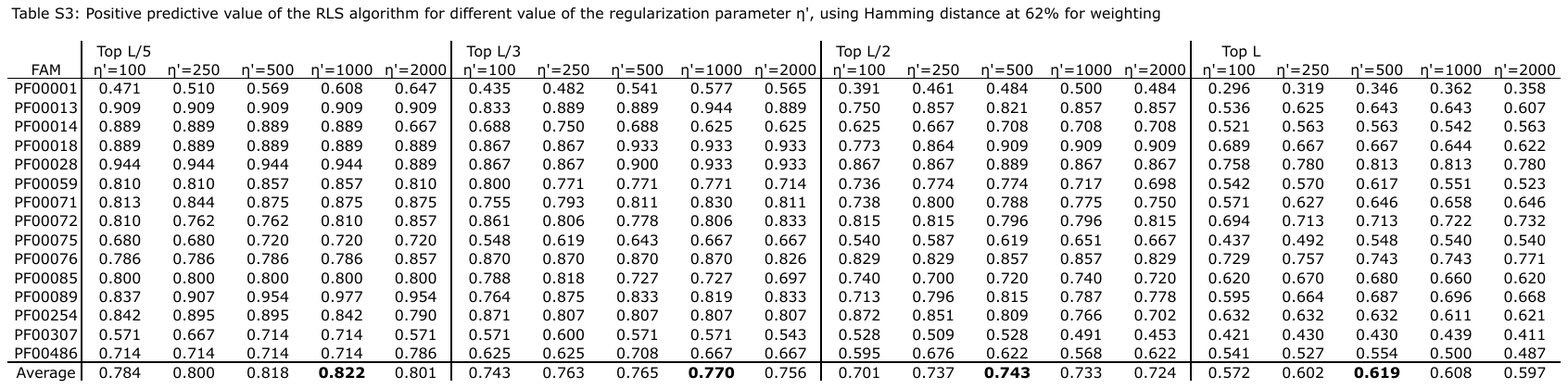}
\includepdf[pages=1,landscape=false]{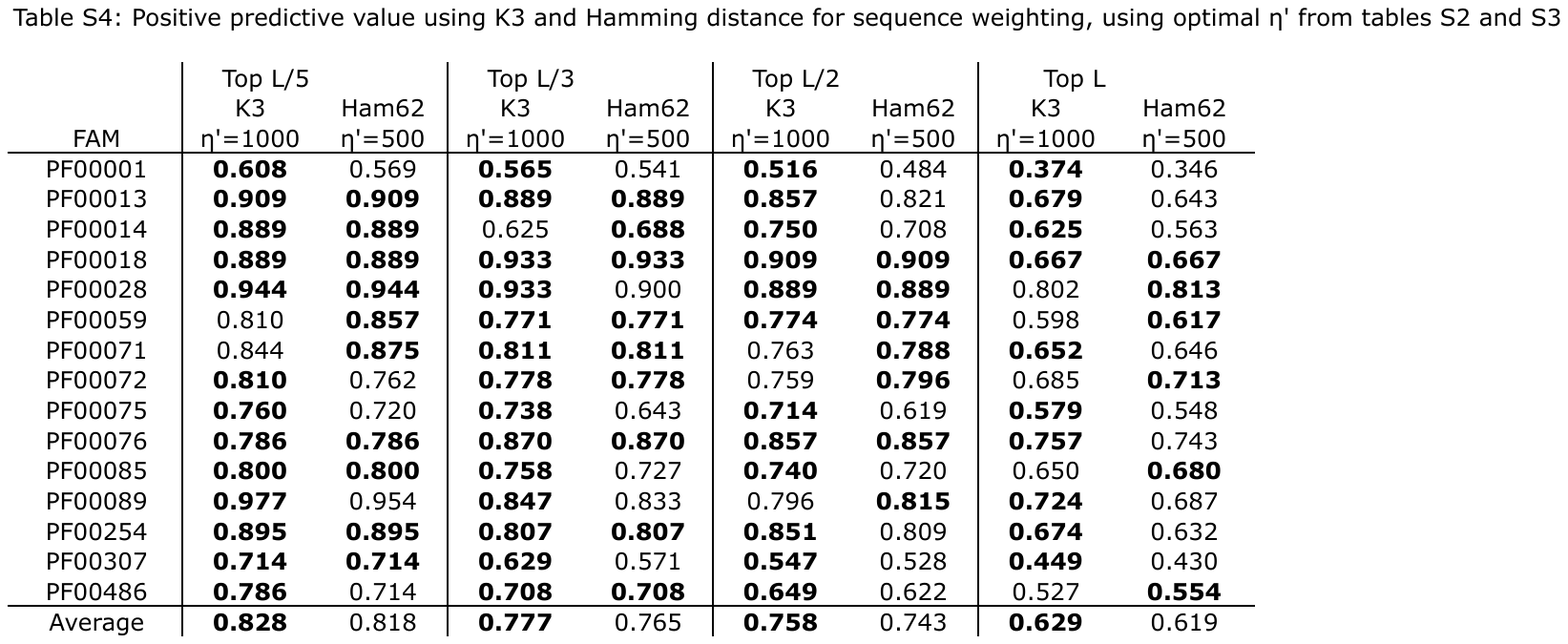}
\includepdf[pages=1,landscape=false]{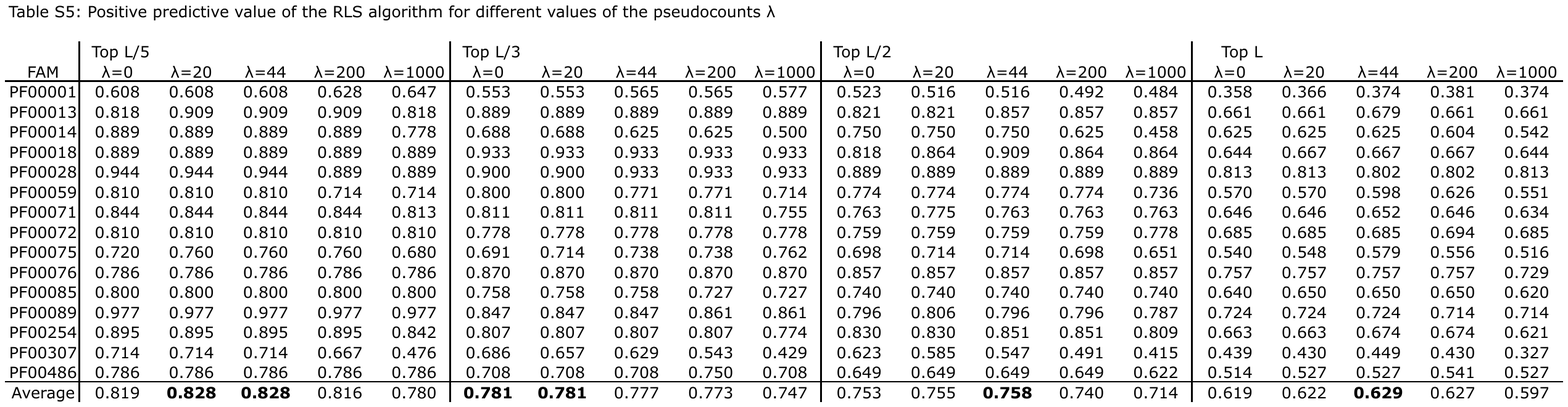}
\includepdf[pages=1,landscape=false]{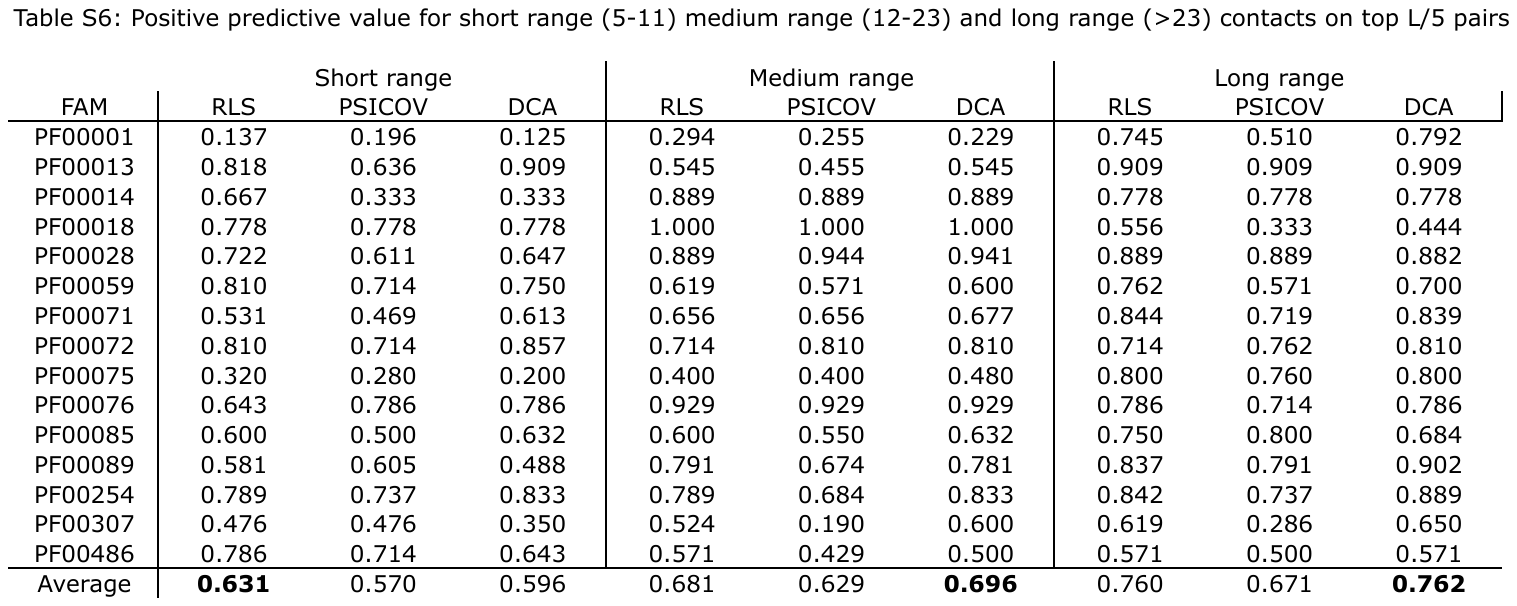}

\end{document}